\documentclass[11pt]{article}

\usepackage[margin=0.9in]{geometry}
\usepackage{amsmath,amssymb,amsthm,mathtools}
\usepackage{booktabs}
\usepackage{xcolor}
\usepackage{enumitem}
\usepackage{algorithm}
\usepackage{algpseudocode}
\algrenewcommand{\algorithmicrequire}{\textbf{Input:}}
\algrenewcommand{\algorithmicensure}{\textbf{Output:}}
\usepackage{microtype}
\usepackage[hidelinks]{hyperref}
\usepackage{url}

\allowdisplaybreaks

\newtheorem{theorem}{Theorem}[section]
\newtheorem{lemma}{Lemma}[section]
\newtheorem{proposition}{Proposition}[section]

\theoremstyle{remark}

\DeclareMathOperator{\OPT}{OPT}
\DeclareMathOperator{\rank}{rank}
\newcommand{\cl}{\operatorname{cl}}
\newcommand{\eps}{\varepsilon}

\newcommand{\Insert}{\mathsf{Insert}}
\newcommand{\Remove}{\mathsf{Remove}}
\newcommand{\staridx}{\ast}

\title{Multilevel Dynamic Thinning for Matroid Secretary}
\author{
Dennis Joyce
}
\date{}

\begin{document}
\maketitle
\begingroup
\renewcommand{\thefootnote}{}
\begin{NoHyper}
\footnotetext{East Texas A\&M University. djoyce2@leomail.etamu.edu}
\end{NoHyper}
\endgroup

\begin{abstract}
In the matroid secretary problem, weighted elements arrive in random order, and an online algorithm must irrevocably accept elements forming a high-weight independent set.  
Dynamic Thinning is a recent, conceptually simple $3.1462$-competitive algorithm for the matroid secretary problem that maintains a random reference set.  Given this reference set, the elements in its maximum-weight independent subset have been accepted independently with a time-dependent probability.  We extend this approach by replacing the single reference set with a finite hierarchy of nested reference sets.
   For every fixed $\eps>0$, the resulting algorithm is $(e+\eps)$-probability-competitive for arbitrary matroids, with $O_\eps(n^2)$ queries in the worst case.  
\end{abstract}

\section{Introduction}

 The matroid secretary problem studies online acceptance under a matroid constraint.  The elements have fixed weights but arrive in uniformly random order.  The algorithm sees each element only when it arrives and must immediately decide whether to accept it.  A probability-competitive ratio $C$ means that every element of a fixed maximum-weight independent set is accepted with probability at least $1/C$.  The classical secretary problem is the rank-one case, where the optimal probability of accepting the best element is $1/e$.  

For arbitrary matroids, the early guarantees depended on the rank $r$, the maximum size of an independent set.  Babaioff et al.~\cite{babaioff2007,babaiokempe2018} obtained a competitive ratio 
of $O(\log r)$.  Later work reduced the rank dependence first to $O(\sqrt{\log r})$~\cite{chakraborty2012} and then to 
$O(\log\log r)$~\cite{lachish2014,feldman2018}.  Constant competitive ratios were already known for structured families such as uniform~\cite{kleinberg2005,babaioff2007}, graphic~\cite{banihashem2025,dutting2026}, and laminar matroids~\cite{imwang2011,huanglaminar2024}.

Recent work has substantially changed the landscape.  Singla~\cite{singla2026} gave a $4$-probability-competitive algorithm for arbitrary matroids.  Bérczi et al.~\cite{berczi2026} and Abdi et al.~\cite{abdi2026} independently obtained the exact $e$-probability-competitive ratio for linear matroids.  
Huang~\cite{huang2026} extended the $e$-probability-competitive guarantee to arbitrary matroids using an expected polynomial number of independence queries.  These breakthroughs leave a useful design question: can a simple algorithm approach the $e$-probability-competitive ratio with fewer independence queries?

Dynamic Thinning is an ordinal algorithm, meaning that it uses only comparisons among observed weights.  The algorithm, given by Chan et al.~\cite{clw2026}, has probability-competitive ratio $3.1462$ and uses a quadratic number of independence queries.  It first rejects an initial sample and uses it to form a \emph{reference set} $X$.  This is an auxiliary set of arrived elements, distinct from the set $A$ of irrevocably accepted elements.  When an element $e$ arrives, the algorithm checks whether $e$ belongs to the greedy set of $X\cup\{e\}$ (its maximum-weight independent subset, with ties broken consistently).  If not, it rejects $e$.  If $e$ can be added to the greedy set of $X$ while keeping all its existing elements, it is accepted with a probability that depends on the arrival time.  If $e$ instead replaces an element $f$ in that greedy set, it is accepted when $f\notin A$ and rejected otherwise.  An accepted element joins both $X$ and $A$.  Meanwhile, unaccepted elements in the greedy set of $X$ are randomly removed from $X$, leaving room for later arrivals.  This algorithm builds on Singla's framework of reversible reference updates to preserve a configuration distribution~\cite{singla2026}.  Its conditional-thinning invariant is motivated by Abdi et al.~\cite{abdi2026}: given the current reference, processed elements of its greedy basis are accepted independently with probability $1/2$.  Dynamic Thinning makes this probability time-dependent.

We extend Dynamic Thinning by replacing its single reference set with a finite reference hierarchy
\[
X_0\subseteq X_1\subseteq\cdots\subseteq X_j
\]
during stage $j$, where there are $m$ stages in total.  The accepted elements always belong to the greedy set of $X_0$.  At predetermined times, the algorithm adds a new outer reference set containing all elements that have arrived so far.  We call this operation a \emph{refresh}.  An element rejected earlier may therefore re-enter the reference hierarchy later.  The additional reference sets allow more arrived elements to be used as references, which can improve acceptance probabilities. With only one stage ($m=1$), our algorithm reduces exactly to Dynamic Thinning.

\paragraph{Main result.}
For every fixed $\eps>0$, we obtain an ordinal $(e+\eps)$-probability-competitive algorithm for arbitrary matroids.  When the refresh times form a geometric progression and the number of stages is $m=\lceil\frac{e}{\epsilon}\rceil$, every element of the fixed tie-broken optimum is accepted with probability at least $1/(e+\eps)$. 

\begin{theorem}[Near-$e$ Multilevel Dynamic Thinning]\label{thm:main}
For every $\eps>0$, there exists an ordinal $(e+\eps)$-probability-competitive algorithm for arbitrary matroids.  It knows only the number $n$ of elements and uses $O\bigl(\lceil e/\eps\rceil 2^{\lceil e/\eps\rceil}n^2\bigr)=O_\eps(n^2)$ independence queries in the worst case, all on already-arrived elements.
\end{theorem}

Huang~\cite{huang2026} gives an exact \(e\)-probability-competitive algorithm for arbitrary matroids using a continuous survival process that tracks which sampled elements remain available.  Its expected running time and expected number of independence queries are polynomial. 
Our algorithm is essentially a finite-level discretization of Huang's survival-process construction: the nested reference sets replace his continuous filtration by finitely many levels (i.e., stages).  For each fixed instance, under our geometric refresh times, the distribution of the final accepted set converges to Huang's output distribution as the number $m$ of levels tends to infinity. 
 Our algorithm updates the reference hierarchy using insertion and removal rules: their design, including how they are paired to maintain the prescribed joint distribution, is derived directly from Huang's forward and reverse kernel construction~\cite{huang2026}.
 More details are discussed in Section~\ref{sec:related}.

The paper is organized as follows.
Section~\ref{sec:prelim} introduces the model and the greedy facts used throughout.  Section~\ref{sec:multilevel} defines the refresh times and the nested reference sets.
 Sections~\ref{sec:local} and~\ref{sec:algorithm} give the update rules and the online algorithm, respectively.  Section~\ref{sec:distribution} proves the joint distribution, and Section~\ref{sec:ratio} derives the near-$e$ guarantee and the query bound.  Section~\ref{sec:related} compares our algorithm with prior work.

\section{Preliminaries}\label{sec:prelim}

\subsection{Matroid secretary problem}

Let $E$ contain $n$ elements, each with a fixed positive weight $w_e$.  Their arrival order is uniformly random, and an element's weight becomes known only when it arrives.  An online decision to accept or reject an element cannot be revised.  The algorithm knows $n$ but learns feasibility only by asking whether sets of already-arrived elements are independent.  Its accepted set must be independent after every decision.
Feasibility is specified by a matroid $M=(E,\mathcal I)$.  For $S\subseteq E$, let $\rank(S)$ be the maximum size of an independent subset of $S$.  We write $S+e=S\cup\{e\}$ and $S-f=S\setminus\{f\}$.  The closure
\[
\cl(S)=\{e\in E:\rank(S+e)=\rank(S)\}
\]
records which elements cannot increase the rank of $S$.  In this case we also say that $S$ spans $e$.

To make the optimum unambiguous when weights tie, fix a label order independently of the arrivals.  Write $x\succ y$ when $x$ precedes $y$ in decreasing weight-and-label order.  Write $\OPT(S)$ for the independent set obtained by considering the positive-weight elements of $S$ in this order and adding each element that preserves independence.  We call $\OPT(S)$ the \emph{greedy set} of $S$.  Matroid greedy optimality makes it a maximum-weight independent subset of $S$.  In particular, $O=\OPT(E)$ is the canonical optimum used in our guarantee.  If $A$ is the random accepted set, a $C$-\emph{probability-competitive} algorithm satisfies
\[
\Pr[e\in A]\ge \frac1C\qquad(e\in O).
\]
Since weights are nonnegative, summing these inequalities with weights $w_e$ also gives the usual $C$-competitive expected-weight bound.

For timing purposes, draw $n$ independent uniform points in $(0,1)$ and assign them, in increasing order, to the successive arrivals.  The resulting timed process has the same random arrival order, but lets us specify sampling thresholds and internal updates between arrivals.  Drawing the points supplies no information about future elements.

\subsection{Greedy facts for the updates}

The update rules use three elementary consequences of matroid exchange.  The first two greedy-set facts were proved by Singla~\cite{singla2026} and are also stated in the reference-set language used by Chan et al.~\cite{clw2026}.  We record them without proof.
\begin{lemma}[Inheritance of a canonical optimum]\label{lem:inherit}
For every $e\in O$ and every $S\subseteq E\setminus\{e\}$,
\[
e\in\OPT(S+e).
\]
More generally, if $e\in\OPT(T)$ and $e\in S\subseteq T$, then $e\in\OPT(S)$.
\end{lemma}

\begin{lemma}[Single greedy exchange]\label{lem:exchange}
Let $S\subseteq E\setminus\{e\}$ and suppose that $e\in\OPT(S+e)$.  If $e\notin\cl(S)$, then
\[
\OPT(S+e)=\OPT(S)+e.
\]
Otherwise there is a unique $f\in\OPT(S)$ such that
\[
\OPT(S+e)=\OPT(S)-f+e.
\]
Moreover, deleting a member $f\in\OPT(S)$ preserves every other member of that greedy set:
\[
\OPT(S)-f\subseteq\OPT(S-f).
\]
\end{lemma}

We will repeatedly use the following levelwise form of the closure exchange property.  

\begin{lemma}[Levelwise closure exchange]\label{lem:layerexchange}
Let $F=(F_0,\ldots,F_j)$ with $F_0\subseteq\cdots\subseteq F_j$, let $f\in E\setminus F_j$, and let $e\in E\setminus\{f\}$.   For $z\in\{e,f\}$, write
\[
\sigma_F(z)=\min\{k\in\{0,\ldots,j\}:z\in\cl(F_k)\},
\]
with value $\staridx>j$ if no level spans $z$.   Let $F+e=(F_0+e,\ldots,F_j+e)$, and form $F^{i,f}$ from the original $F$ by adding $f$ only to levels $k\ge i$.
If $\sigma_{F+e}(f)<\sigma_F(f)$, then for every $i\in\{0,\ldots,j\}$,
\[
\begin{aligned}
\sigma_{F^{i,f}}(e)&=\min\bigl\{\sigma_F(f),\max\{\sigma_{F+e}(f),i\}\bigr\},\\
\sigma_{F^{i,f}+e}(f)&=\min\{\sigma_{F+e}(f),i\}.
\end{aligned}
\]
\end{lemma}

\begin{proof}
At each level $k$, closure exchange gives
\[
e\in\cl(F_k+f)\setminus\cl(F_k)
\quad\Longleftrightarrow\quad
f\in\cl(F_k+e)\setminus\cl(F_k).
\]
The right-hand side holds exactly when $\sigma_{F+e}(f)\le k<\sigma_F(f)$, so $e\notin\cl(F_k)$ for every such $k$.  By monotonicity of closure, earlier levels do not span $e$ either.  Thus no $F_k$ with $k<\sigma_F(f)$ spans $e$.

If $\sigma_F(f)\le j$, apply closure exchange at $k=\sigma_{F+e}(f)$.  Since $F_{\sigma_{F+e}(f)}\subseteq F_{\sigma_F(f)}$ and $f\in\cl(F_{\sigma_F(f)})$, monotonicity and idempotence of closure give
\[
e\in\cl\bigl(F_{\sigma_{F+e}(f)}+f\bigr)
\subseteq\cl\bigl(F_{\sigma_F(f)}\bigr).
\]
Thus $F_{\sigma_F(f)}$ spans $e$, whereas every earlier level does not.  By the definition of the first spanning level, $\sigma_F(e)=\sigma_F(f)$.
If $\sigma_F(f)=\staridx$, every actual level $k\le j$ satisfies $k<\sigma_F(f)$, so no level of $F$ spans $e$.  The definition then gives again $\sigma_F(e)=\sigma_F(f)$.

For the first identity, levels below $i$ are unchanged in $F^{i,f}$.  At levels $k\ge i$, adding $f$ newly spans $e$ exactly when $\sigma_{F+e}(f)\le k<\sigma_F(f)$, by the equivalence above.  The earliest such level is $\max\{\sigma_{F+e}(f),i\}$ if it precedes $\sigma_F(f)$.  Otherwise, $e$ retains its original span index $\sigma_F(f)$.  This proves the stated minimum formula.

For the second identity, the sets at levels $k<i$ in $F^{i,f}+e$ are $F_k+e$, which span $f$ exactly when $k\ge\sigma_{F+e}(f)$.  Every level $k\ge i$ contains $f$ itself.  Hence the first level spanning $f$ is $\min\{\sigma_{F+e}(f),i\}$.
\end{proof}

\section{Reference Hierarchy}\label{sec:multilevel}

Dynamic Thinning uses a single random reference set $X$.  Conditional on $X$, the elements of $\OPT(X)$ have been accepted independently with a common time-dependent probability.  Our multilevel algorithm extends this conditional-thinning description to a nested hierarchy of reference sets.  The additional levels allow rejected elements to return to the reference hierarchy.

Fix an integer $m\ge1$ and a deterministic time schedule
\[
0<T_0=p<T_1<\cdots<T_{m-1}<T_m=1.
\]
Here $T_0=p$ ends the initial sampling period, $T_1,\ldots,T_{m-1}$ are the \emph{refresh times}, and $T_m=1$ ends the execution. We assume $T_0=p\ge \frac1e$ throughout the paper.  There are $m$ stages.  Stage $j$ is the interval $[T_j,T_{j+1})$ for $0\le j<m$.  At the beginning of stage $j$, the algorithm has the nested reference sets
\[
X_0\subseteq X_1\subseteq\cdots\subseteq X_j.
\]
At each refresh time $T_j$ with $1\le j<m$, the algorithm keeps the existing reference sets and adds a new outer set $X_j$ containing all elements that have arrived by $T_j$.  The innermost set $X_0$ plays the role of the single reference set in Dynamic Thinning.

We call an element \emph{active} if it is currently stored in at least one reference set.  During stage $j$, these are precisely the elements of $X_j$, because the reference sets are nested.  The \emph{level} of an active element $f$ is the smallest $i$ such that $f\in X_i$.  Thus a level-$i$ element belongs to $X_i,X_{i+1},\ldots,X_j$, but not to the earlier reference sets.  The level-$0$ elements form $X_0$.
Let $A$ denote the set of elements accepted online. 

Our design goal is to maintain a conditional distribution of $A$ given the reference hierarchy.  At time $t$ in stage $j$, we require
\[
A\subseteq\OPT(X_0),
\]
and, conditional on $X_0,\ldots,X_j$, the events $\{f\in A\}$ for $f\in\OPT(X_0)$ should be mutually independent.  Unlike in the single-reference-set algorithm, their probabilities need not be equal: the target \emph{conditional acceptance probability} $\Pr[f\in A\mid X_0,\ldots,X_j]$ depends on the first reference level at which $f$ leaves the greedy set.
In particular, if 
\[
f\in\OPT(X_0)\cap\cdots\cap\OPT(X_{i-1})
\quad\text{but}\quad
f\notin\OPT(X_i),
\]
the target probability is $g(T_i)/p$, where $g$ is a function defined in \eqref{eq:qgdef}.  If $f$ remains in $\OPT(X_j)$, the target probability is $g(t)/p$.
These are probabilities of having already been accepted, conditional on the current reference sets, rather than probabilities used directly at an element's arrival.  Our task is to design online updates of the reference sets that maintain this conditional distribution of $A$ while leaving all previous acceptance and rejection decisions unchanged.

To define $g$ from the refresh schedule, let $b(t)$ record the beginning of the current stage:
\[
b(t)=T_j\qquad\text{for }T_j\le t<T_{j+1}.
\]
Then define
\begin{equation}\label{eq:qgdef}
q(t)=\exp\left(-\int_p^t\frac{ds}{b(s)}\right),
\qquad
g(t)=\int_p^t q(s)\,ds.
\end{equation}
They satisfy
\begin{equation}\label{eq:qgdiff}
q'(t)=-\frac{q(t)}{T_j},
\qquad
g'(t)=q(t)
\qquad (T_j<t<T_{j+1}).
\end{equation}
In particular, within a stage,
\begin{equation}\label{eq:qgstage}
q(t)=q(T_j)e^{-(t-T_j)/T_j},
\qquad
g(t)=g(T_j)+T_jq(T_j)\bigl(1-e^{-(t-T_j)/T_j}\bigr).
\end{equation}

The restriction $p\ge1/e$ ensures that all quantities used later as probabilities or rates are nonnegative.

\begin{lemma}[Validity of the schedule]\label{lem:schedulevalid}
For every $t\in(p,1]$,
\[
0<q(t)\le \frac pt,
\qquad
0\le g(t)<p.
\]
\end{lemma}

\begin{proof}
Since $b(s)\le s$,
\[
q(t)=\exp\left(-\int_p^t\frac{ds}{b(s)}\right)
\le
\exp\left(-\int_p^t\frac{ds}{s}\right)
=\frac pt.
\]
It follows that
\[
g(t)\le p\log\frac tp\le p\log\frac1p\le p.
\]
The first inequality is strict after integration whenever $t>p$, because the finite step function $b(s)$ is strictly smaller than $s$ except at the stage endpoints.  Hence $g(t)<p$.
\end{proof}

We also write
\[
q_i=q(T_i),\quad g_i=g(T_i)\quad(0\le i\le m),\qquad d_i=T_i-T_{i-1}\quad(1\le i\le m),
\]
and set $d_0=p$.

Fix a current time $t\in[T_j,T_{j+1})$ in stage $j$.  Use the ordered indices
\[
0<1<\cdots<j<\staridx.
\]
and set
\[
q_{\staridx}=q(t),\qquad g_{\staridx}=g(t).
\]
For a nested family $F=(F_0,\ldots,F_j)$ and an element $z\in E$, the convention $\sigma_F(z)=\staridx$ means $z\notin\cl(F_j)$, as in Lemma~\ref{lem:layerexchange}.  The values $q_{\staridx}$ and $g_{\staridx}$ are therefore evaluated at the current time $t$ whenever this span index occurs.

The definitions of $q$ and $g$ above imply the following identity involving $q_i$, $g_i$, and $d_i$.  We will use this identity to define valid probabilities for updating the reference sets and to prove that the resulting online algorithm maintains the target conditional distribution of $A$.

\begin{lemma}[Finite-level balance]\label{lem:scalarbalance}
Fix a stage $j$ and a time $t\in[T_j,T_{j+1})$.  For any $a,b\in\{0,\ldots,j,\staridx\}$ with $a>b$,
\begin{equation}\label{eq:balanceidentity}
(p-g_a)q_b+
\sum_{\substack{1\le i<a\\i\le j}}
 d_i\bigl(q_{\max\{b,i\}}-q_a\bigr)q_{\min\{b,i\}}
=q_a(p-g_b).
\end{equation}
\end{lemma}

\begin{proof}
For $k\in\{0,1,\ldots,j\}$, define
\[
L_k=\sum_{i=1}^k d_iq_i,
\]
and set $L_{\staridx}=L_j$.  Also define $\beta_k=T_k$ for $k\le j$ and $\beta_{\staridx}=T_j$.  We first establish the auxiliary identity
\[
p-g_k+L_k=\beta_kq_k
\qquad(k\in\{0,\ldots,j,\staridx\}).
\]
At $t=p$, both sides equal $p$.  Within a stage, only the identity for $k=\staridx$ varies with time, and
\[
\frac{d}{dt}(p-g(t))=-q(t)=\frac{d}{dt}\bigl(T_jq(t)\bigr).
\]
At a refresh time $T_i$, the left side gains the term $d_iq_i$, while the right side changes from $T_{i-1}q_i$ to $T_iq_i$.  Thus the auxiliary identity holds throughout the schedule, including at each refresh time where the values for a finite index $k$ are defined.

For~\eqref{eq:balanceidentity}, split the sum at $b$.  Since the term $i=a$ would be zero, the sum can be written as
\[
q_bL_a-q_aL_b-q_aq_b(\beta_a-\beta_b).
\]
Subtract the right side of~\eqref{eq:balanceidentity} from the left side and substitute this expression.  The result is
\[
q_b(p-g_a+L_a)-q_a(p-g_b+L_b)-q_aq_b(\beta_a-\beta_b),
\]
which is zero after substituting $p-g_a+L_a=\beta_aq_a$ and $p-g_b+L_b=\beta_bq_b$.
\end{proof}


\section{Updating the Reference Sets}\label{sec:local}

We define two local randomized procedures for updating reference sets when an element is inserted/added to or removed from all current reference sets.  These procedures adjust the levels of other elements without changing their acceptance status.  The insertion or removal of the element itself is carried out by the online algorithm in Section~\ref{sec:algorithm}.  

The design of the insertion and removal rules in Sections~\ref{sec:insertion-scan} and~\ref{sec:removal-scan}, including the method of pairing their recursive choices to maintain the joint distribution, is taken directly from the forward and reverse kernel construction by Huang~\cite{huang2026}.  We express this pairing in terms of finite level assignments and their weights. Lemma~\ref{lem:localflow} gives the corresponding probability-flow identity.

\subsection{Backgrounds, span indices, and local weights}

Fix a stage $j$ and a current time $t$.  Each local procedure processes elements in the fixed decreasing weight-and-label order.  A \emph{background} is a nested sequence
\[
F=(F_0,\ldots,F_j),
\qquad F_0\subseteq\cdots\subseteq F_j,
\]
where, when processing an element $f$, $F_i$ consists of the previously processed elements that precede $f$ in this order and belong to the $i$th reference set in the configuration considered by the current recursive call.  Here, processing refers to the local procedure, not to online arrival.  For an element $z$ and an index $i\in\{0,\ldots,j\}$, let $F^{i,z}$ be the background obtained by adding $z$ to $F_i,F_{i+1},\ldots,F_j$ and leaving the earlier sets unchanged.  Write $F+z=F^{0,z}$.

Let $e$ be the element being inserted or removed, and let $L\subseteq E\setminus F_j$ be a finite set of elements not yet processed by the local procedure, all following $e$ in the fixed weight-and-label order.  A \emph{level assignment} $c=(c_f)_{f\in L}$ specifies a possible level $c_f\in\{0,\ldots,j\}$ for each element $f\in L$.  To apply a level assignment $r$ to the reference sets means to adjust their membership so that, for every element $f$ in the domain of $r$ and every $0\le i\le j$,
\[
f\in X_i \quad\Longleftrightarrow\quad r_f\le i,
\]
while leaving the membership of all elements outside the domain of $r$ unchanged.  Thus applying $r$ may change the inner reference sets $X_0,\ldots,X_{j-1}$ even though it does not change whether an element in the domain of $r$ belongs to $X_j$.  Write $A_L=A\cap L$ for the accepted elements in $L$.  We will restrict that $A_L$ is valid only if every element of $A_L$ has level $0$ and belongs to the resulting level-$0$ greedy set.

The local weight $W_F(L,c,A_L)$ assigns a nonnegative weight to the level assignment $c$ and the accepted restriction $A_L=A\cap L$, given the background $F$.  It is defined recursively by processing the elements of $L$ in decreasing weight-and-label order.

For the empty set, define
\[
W_F(\varnothing,\varnothing,\varnothing)=1.
\]
For nonempty $L$, let $f$ be its largest element.  Write $c|_{L-f}$ for the restriction of the level assignment to the remaining elements.  Define
\begin{equation}\label{eq:localweight}
W_F(L,c,A_L)=
\begin{cases}
g_{\sigma_F(f)}\,W_{F+f}(L-f,c|_{L-f},A_L\setminus\{f\}),
  &c_f=0,\ f\in A_L,\\[3pt]
\bigl(p-g_{\sigma_F(f)}\bigr)\,W_{F+f}(L-f,c|_{L-f},A_L\setminus\{f\}),
  &c_f=0,\ f\notin A_L,\\[3pt]
d_i\,W_{F^{i,f}}(L-f,c|_{L-f},A_L\setminus\{f\}),
  &c_f=i\ge1,\ f\notin A_L,\\[3pt]
0,
  &c_f\ge1,\ f\in A_L.
\end{cases}
\end{equation}
In the first two cases, $f$ has level $0$, so it is added to every background set.  Its factor is $g_{\sigma_F(f)}$ if accepted and $p-g_{\sigma_F(f)}$ otherwise.  In the third case, $f$ is unaccepted and has level $i\ge1$, so its factor is $d_i$ and it is added only to $F_i,\ldots,F_j$.  Every recursive call removes $f$ from both $L$ and the accepted subset passed to that call.  In the second and third cases, $f\notin A_L$, so $A_L\setminus\{f\}=A_L$.  This removal occurs only in the temporary arguments of the weight calculation and is not an online rejection.

The last case assigns weight zero to an accepted element outside level $0$.  If an accepted level-$0$ element is spanned by $F_0$, then $\sigma_F(f)=0$ and its factor is also zero because $g_0=0$.  Thus any configuration violating the acceptance conditions has weight zero.  Each nonzero recursive call replaces $L$ by $L-f$, so $|L|$ strictly decreases and the recursion reaches the empty-set case after finitely many steps. 

Summing these weights over all level assignments and accepted subsets gives the following normalization identity.

\begin{lemma}[Local normalization]\label{lem:localnorm}
For every background $F$ and every finite $L\subseteq E\setminus F_j$,
\[
\sum_{c,A_L}W_F(L,c,A_L)=T_j^{|L|},
\]
where the sum ranges over all level assignments and accepted subsets.
\end{lemma}

\begin{proof}
We proceed by induction on $|L|$.  The claim is immediate for $L=\varnothing$.  Suppose that $L\neq\varnothing$, and let $f$ be its largest element.  Write $\sigma=\sigma_F(f)$.

For each possible level and acceptance status of $f$, the recursive call is made with $L-f$ and some updated background $F'$.  By the induction hypothesis, which holds for every background,
\[
\sum_{c',A'}W_{F'}(L-f,c',A')=T_j^{|L|-1}.
\]
Hence, summing first over the level and acceptance status of $f$, we obtain
\[
\begin{aligned}
\sum_{c,A_L}W_F(L,c,A_L)
&=
\left(
g_\sigma+(p-g_\sigma)+\sum_{i=1}^j d_i
\right)T_j^{|L|-1}\\
&=
\left(
p+\sum_{i=1}^j(T_i-T_{i-1})
\right)T_j^{|L|-1}\\
&=T_j^{|L|}.
\end{aligned}\qedhere
\]
\end{proof}

\subsection{Insertion scan}\label{sec:insertion-scan}

The procedure $\Insert(L,F,e,c,A_L)$ randomly constructs the level changes needed for a proposed addition of $e$ to all current reference sets $X_0,\ldots,X_j$.  Its output is either failure or a new level assignment for $L$.  The online algorithm calls this procedure when $e$ arrives.  If the call fails, it rejects $e$ and leaves the reference sets unchanged.  If the call succeeds, it updates the levels of the elements in $L$ as returned by the procedure, adds $e$ to all current reference sets, and accepts $e$.  The acceptance status of every element already in $L$ remains unchanged.  The procedure also calls itself to work out the level changes for smaller elements. 
Algorithm~\ref{alg:insert} gives the complete recursive definition.

Algorithm~\ref{alg:insert} may return before processing an element of $L$ in two situations.  If $\sigma_F(e)=0$ (lines~\ref{line:insert-spanned-check}--\ref{line:insert-spanned-return}), the innermost background already spans $e$, so adding $e$ does not require any level change.  If $L=\varnothing$ (lines~\ref{line:insert-empty-check}--\ref{line:insert-empty-return}), no smaller element remains, and the procedure succeeds with probability $q_{\sigma_F(e)}$.

Suppose $\sigma_F(e)\ne0$ and $L\ne\varnothing$, so the procedure processes the largest element $f$ of $L$.  When a recursive call returns an assignment on $L-f$, extending it by $c_f=i$ means assigning level $i$ to $f$ and leaving all returned levels unchanged.

Here $a=\sigma_F(f)$ is the first level whose current background spans $f$, while $b=\sigma_{F+e}(f)$ is the corresponding level after adding $e$.  Monotonicity gives $a\ge b$.  If $a=b$, the level of $f$ is unchanged and the procedure continues with the smaller elements.  Suppose $a>b$.  In (I1), $f$ is already accepted at level $0$, so the coin with success probability $\frac{q_ag_b}{q_bg_a}$ prevents the update from assigning too much probability to this accepted configuration.  In (I2), $f$ is unaccepted at level $0$ and can remain there without this additional coin.

In (I3), $f$ is an unaccepted element at a later level $i$.  Lemma~\ref{lem:layerexchange} gives
\[
s=\sigma_{F^{i,f}}(e)=\min\{a,\max\{b,i\}\}.
\]
With probability $q_a/q_s$, the procedure keeps $f$ at level $i$.  Otherwise it first inserts $e$ with $f$ still at level $i$, then inserts $f$ again to move it from level $i$ to level $0$.  We call this alternative the \emph{recycling step}.  The recursive calls may also change the levels of smaller elements, but they leave the acceptance status of $f$ and all those elements unchanged.  Since $f$ was already rejected online, it remains outside the accepted set.

\begin{algorithm}[H]
\caption{Insertion procedure $\Insert(L,F,e,c,A_L)$}
\label{alg:insert}
\begin{algorithmic}[1]
\Require A background $F$, an element $e$ whose addition is being considered, the remaining set $L$, its level assignment $c$, and its accepted subset $A_L$
\Ensure Either failure or a new level assignment for $L$
\If{$\sigma_F(e)=0$}\label{line:insert-spanned-check}
  \State \Return $c$\label{line:insert-spanned-return}
\EndIf
\If{$L=\varnothing$}\label{line:insert-empty-check}
  \State \Return the empty assignment with probability $q_{\sigma_F(e)}$, and failure otherwise\label{line:insert-empty-return}
\EndIf
\State Let $f$ be the largest element of $L$
\State $a\gets\sigma_F(f)$, \quad $b\gets\sigma_{F+e}(f)$
\If{$a=b$}
  \State $F'\gets F+f$ if $c_f=0$, and $F'\gets F^{c_f,f}$ otherwise
  \State $r\gets\Insert(L-f,F',e,c|_{L-f},A_L\setminus\{f\})$
  \State \Return failure if $r$ is failure, and otherwise return $r$ extended by $c_f$
\EndIf
\If{$c_f=0$ and $f\in A_L$} \Comment{(I1)}
  \State With probability $1-\frac{q_ag_b}{q_bg_a}$, \Return failure
  \State $r\gets\Insert(L-f,F+f,e,c|_{L-f},A_L\setminus\{f\})$
  \State \Return failure if $r$ is failure, and otherwise return $r$ extended by $c_f=0$
\ElsIf{$c_f=0$ and $f\notin A_L$} \Comment{(I2)}
  \State $r\gets\Insert(L-f,F+f,e,c|_{L-f},A_L\setminus\{f\})$
  \State \Return failure if $r$ is failure, and otherwise return $r$ extended by $c_f=0$
\Else \Comment{(I3): $c_f=i\ge1$}
  \State $i\gets c_f$
  \State $s\gets\sigma_{F^{i,f}}(e)=\min\{a,\max\{b,i\}\}$
  \If{a coin with success probability $q_a/q_s$ returns success}
    \State $r\gets\Insert(L-f,F^{i,f},e,c|_{L-f},A_L\setminus\{f\})$
    \State \Return failure if $r$ is failure, and otherwise return $r$ extended by $c_f=i$
  \Else
    \State $r_1\gets\Insert(L-f,F^{i,f},e,c|_{L-f},A_L\setminus\{f\})$
    \If{$r_1$ is failure}
      \State \Return failure
    \EndIf
    \State $r_2\gets\Insert(L-f,F^{i,f}+e,f,r_1,A_L\setminus\{f\})$
    \State \Return failure if $r_2$ is failure, and otherwise return $r_2$ extended by $c_f=0$
  \EndIf
\EndIf
\end{algorithmic}
\end{algorithm}

\subsection{Removal scan}\label{sec:removal-scan}

The procedure $\Remove(L,F,e,c',A_L)$ is a randomized update used when $e$ is removed from all current reference sets.  Given the current assignment $c'$, it samples a new level assignment for $L$ so that the required distribution is preserved after the removal of $e$. The sampling probabilities of $\Remove$ choose among these possible assignments and sum to one, so the procedure never fails.  Its output need not equal the level assignment that occurred at some earlier time.

When the online algorithm calls $\Remove$, its input $e$ is an unaccepted reference element selected by the rule in Section~\ref{sec:deletions}.  After the procedure returns, the online algorithm adjusts the reference sets according to the returned level assignment, and removes $e$ from all current reference sets.  During the recursive computation, $\Remove$ may move a previously rejected element from level $0$ back to a later level $i$.

Algorithm~\ref{alg:remove} gives the complete recursive definition.  The procedure first checks whether $\sigma_F(e)=0$ or $L=\varnothing$.  In either case, it returns the current assignment $c'$ unchanged: in the first case, the innermost background already spans $e$, so removing $e$ requires no level changes; in the second, no smaller element remains to be processed.

Otherwise, the procedure processes the largest element $f$ of $L$ (line~\ref{line:remove-largest}).  Let $a=\sigma_F(f)$ and $b=\sigma_{F+e}(f)$, as in line~\ref{line:remove-indices}.  These span indices have the same meanings as in Algorithm~\ref{alg:insert}, and monotonicity gives $a\ge b$.  Extending an assignment on $L-f$ by $c_f=i$ assigns level $i$ to $f$ and leaves the returned levels unchanged.

If $a=b$, adding or removing $e$ does not change the relevant span index of $f$, so $\Remove$ keeps the level of $f$ and continues recursively (lines~\ref{line:remove-equal-check}--\ref{line:remove-equal-return}).

Now suppose $a>b$.  If $f$ is an accepted level-$0$ element, it remains at level $0$, and the procedure recursively removes $e$ from the remaining elements with background $F+f$ (lines~\ref{line:remove-accepted-check}--\ref{line:remove-accepted-return}).  If $f$ is currently at a later level $i$, it remains at level $i$, and the procedure recurses with background $F^{i,f}$ (lines~\ref{line:remove-later-check}--\ref{line:remove-later-return}).  In each of these cases, only one output level for $f$ is possible, so no random choice of its level is needed.

It remains to consider the case in which $a>b$ and $f$ is currently unaccepted at level $0$.  Only in this case may several output levels be possible.  The procedure samples its output level $\ell$ from $\{0\}\cup K_a$, where
\[
K_a=\{i\in\{1,\ldots,j\}:i<a\}.
\]
Define
\begin{equation}\label{eq:eta}
\eta=
\frac{(p-g_a)q_b}{q_a(p-g_b)}.
\end{equation}
For each $i\in K_a$, define
\begin{equation}\label{eq:pii}
\pi_i=
\frac{d_i\bigl(q_{\max\{b,i\}}-q_a\bigr)q_{\min\{b,i\}}}
     {q_a(p-g_b)}.
\end{equation}
Lemma~\ref{lem:schedulevalid} shows that these quantities are nonnegative, and Lemma~\ref{lem:scalarbalance} gives
\[
\eta+\sum_{i\in K_a}\pi_i=1,
\]
so they form a probability distribution. The procedure draws $\ell$ with $\Pr[\ell=0]=\eta$ and $\Pr[\ell=i]=\pi_i$ for $i\in K_a$ (line~\ref{line:remove-draw}).
If $\ell=0$, the procedure keeps $f$ at level $0$ and recursively removes $e$ (lines~\ref{line:remove-zero-check}--\ref{line:remove-zero-return}).  If $\ell=i\in K_a$, the procedure first recursively removes $f$ with background $F^{i,f}+e$, then recursively removes $e$ with background $F^{i,f}$, and returns the resulting assignment extended by $c_f=i$ (lines~\ref{line:remove-recycle-level}--\ref{line:remove-recycle-return}).

\begin{algorithm}[H]
\caption{Removal procedure $\Remove(L,F,e,c',A_L)$}
\label{alg:remove}
\begin{algorithmic}[1]
\Require A background $F$, an element $e$ to be removed, the remaining set $L$, its current level assignment $c'$, and its accepted subset $A_L$
\Ensure A new level assignment for $L$ to use after $e$ is removed
\If{$\sigma_F(e)=0$ or $L=\varnothing$}
  \State \Return $c'$
\EndIf
\State Let $f$ be the largest element of $L$\label{line:remove-largest}
\State $a\gets\sigma_F(f)$, \quad $b\gets\sigma_{F+e}(f)$\label{line:remove-indices}
\If{$a=b$}\label{line:remove-equal-check}
  \State $h\gets c'_f$
  \State $F'\gets F+f$ if $h=0$, and $F'\gets F^{h,f}$ otherwise
  \State $r\gets\Remove(L-f,F',e,c'|_{L-f},A_L\setminus\{f\})$
  \State \Return $r$ extended by $c_f=h$\label{line:remove-equal-return}
\EndIf
\If{$c'_f=0$ and $f\in A_L$}\label{line:remove-accepted-check}
  \State $r\gets\Remove(L-f,F+f,e,c'|_{L-f},A_L\setminus\{f\})$
  \State \Return $r$ extended by $c_f=0$\label{line:remove-accepted-return}
\ElsIf{$c'_f\ge1$}\label{line:remove-later-check}
  \State $i\gets c'_f$
  \State $r\gets\Remove(L-f,F^{i,f},e,c'|_{L-f},A_L\setminus\{f\})$
  \State \Return $r$ extended by $c_f=i$\label{line:remove-later-return}
\Else \Comment{$c'_f=0$ and $f\notin A_L$}
  \State Draw $\ell$ with $\Pr[\ell=0]=\eta$ and $\Pr[\ell=i]=\pi_i$ for $i\in K_a$\label{line:remove-draw}
  \If{$\ell=0$}\label{line:remove-zero-check}
    \State $r\gets\Remove(L-f,F+f,e,c'|_{L-f},A_L\setminus\{f\})$
    \State \Return $r$ extended by $c_f=0$\label{line:remove-zero-return}
  \Else
    \State $i\gets\ell$\label{line:remove-recycle-level}
    \State $r_1\gets\Remove(L-f,F^{i,f}+e,f,c'|_{L-f},A_L\setminus\{f\})$
    \State $r_2\gets\Remove(L-f,F^{i,f},e,r_1,A_L\setminus\{f\})$
    \State \Return $r_2$ extended by $c_f=i$\label{line:remove-recycle-return}
  \EndIf
\EndIf
\end{algorithmic}
\end{algorithm}

\subsection{Local probability-flow identity}\label{sec:local-flow}

The next lemma is the main technical component.  It pairs every successful sequence of recursive choices made by $\Insert$ with a sequence of choices made by $\Remove$ having the required probability.  In this identity, an output $c$ of $\Remove$ is interpreted as a possible input to $\Insert$. 

\begin{lemma}[Local flow identity]\label{lem:localflow}
Fix a background $F$, an element $e$ to be inserted, a set $L$ of elements following $e$ in the fixed weight-and-label order, and an accepted subset $A_L$.  Let
\(
I_F(c,c')
\)
be the probability that $\Insert(L,F,e,c,A_L)$ succeeds with output level assignment $c'$, and let
\(
D_F(c',c)
\)
be the probability that $\Remove(L,F,e,c',A_L)$ outputs $c$.  Then
\begin{equation}\label{eq:detailedflow}
W_F(L,c,A_L)I_F(c,c')
=
q_{\sigma_F(e)}
W_{F+e}(L,c',A_L)D_F(c',c).
\end{equation}
Consequently,
\begin{equation}\label{eq:columnflow}
\sum_c W_F(L,c,A_L)I_F(c,c')
=
q_{\sigma_F(e)}W_{F+e}(L,c',A_L).
\end{equation}
\end{lemma}

\begin{proof}
We use induction on $|L|$.  The claim is immediate if $\sigma_F(e)=0$.  If $L=\varnothing$, insertion succeeds with probability $q_{\sigma_F(e)}$ while removal succeeds deterministically, so~\eqref{eq:detailedflow} holds.

Let $f$ be the largest element of a nonempty $L$, and define $a,b$ as in the procedures.  We call $c$ the assignment before insertion and $c'$ the assignment after insertion.  If $a=b$, the local factor of $f$ is identical before and after inserting $e$.  The inserted element $e$ also has the same span index after $f$ is added at its current level.  The induction hypothesis therefore gives the claim.

Assume $a>b$.  We compare the possible levels and acceptance statuses of $f$ after insertion.

If $f$ is accepted at level $0$, its factor before insertion is $g_a$.  The continuation probability $\frac{q_ag_b}{q_bg_a}$ in case~(I1) and the recursive factor $q_b$ multiply to
\[
g_a\cdot\frac{q_ag_b}{q_bg_a}\cdot q_b=q_ag_b,
\]
which is the required factor $g_b$ after insertion, multiplied by $q_a$.

If $f$ remains at a later level $i$, let $s=\sigma_{F^{i,f}}(e)$.  Its level factor is $d_i$ both before and after insertion.  The probability of keeping it there and the recursive identity give
\[
d_i\cdot\frac{q_a}{q_s}\cdot q_s=q_ad_i.
\]

It remains to consider the case in which $f$ is unaccepted at level $0$ after insertion.  One possibility is that it was also unaccepted at level $0$ before insertion.  Its contribution, after applying the induction hypothesis to the remaining elements, is
\[
(p-g_a)q_b.
\]
For each possible earlier level $i<a$, the recycling step contributes
\[
d_i\bigl(q_{\max\{b,i\}}-q_a\bigr)q_{\min\{b,i\}}.
\]
Indeed, $q_{\max\{b,i\}}-q_a$ is the probability from choosing not to keep $f$ at level $i$, after including the factor from the first recursive call.  The second recursive insertion contributes $q_{\min\{b,i\}}$ by Lemma~\ref{lem:layerexchange}.  Summing over all possible earlier levels and applying~\eqref{eq:balanceidentity} gives
\[
q_a(p-g_b),
\]
which is exactly the required weight after insertion.  The probabilities~\eqref{eq:eta} and~\eqref{eq:pii} are the contributions from the possible earlier levels divided by their total weight.  Reversing the recursive calls therefore proves~\eqref{eq:detailedflow} for every paired sequence of recursive choices.

Summing~\eqref{eq:detailedflow} over all assignments $c$ before insertion and using that $\Remove$ always outputs some $c$ gives~\eqref{eq:columnflow}.
\end{proof}

\section{The Online Algorithm}\label{sec:algorithm}

We now place the local insertion and removal scans inside an online process.  Every independence query concerns only already-arrived elements.  Actual acceptance and rejection decisions remain irrevocable.

Algorithm~\ref{alg:online} gives the online Multilevel Dynamic Thinning algorithm.  Let $H_t$ denote the set of elements that have arrived by time $t$.  Besides real arrivals and the fixed refresh times, the algorithm handles \emph{internal-deletion events}.  During each such event, it first runs the removal scan to update the levels of smaller active elements and then removes one rejected element from all current reference sets.  The rejection decision for that element does not change.  In every call to $\Insert$ or $\Remove$, the background $F$, the remaining set $L$, and the level assignment $c$ have the meanings defined in Section~\ref{sec:local}.  In Algorithm~\ref{alg:online}, $c$ denotes the current level assignment determined by the reference sets, and $c|_L$ is its restriction to $L$.  The subsections following the pseudocode explain each part in more detail.

\begin{algorithm}[H]
\caption{Multilevel Dynamic Thinning}
\label{alg:online}
\begin{algorithmic}[1]
\Require A schedule $1/e\le T_0=p<T_1<\cdots<T_m=1$ and the functions $q$ and $g$
\Ensure An accepted set $A$
\State Reject every element arriving by time $p$, and let $S=H_p$ be this sample \label{line:online-sample}
\State $j\gets0$, $X_0\gets S$, $A\gets\varnothing$, and $t\gets p$ \label{line:online-initialize}
\State Set $C_t\gets X_0\cap\OPT(X_j)$, and give each $f\in C_t\setminus A$ an independent clock with rate $q(t)/(p-g(t))$ \label{line:online-clocks}
\While{an arrival, refresh, or clock ring remains before time $1$}
  \State Advance $t$ to the next such event
  \If{the event is the refresh time $T_{j+1}$}
    \State $j\gets j+1$ and $X_j\gets H_t$ \label{line:online-refresh}
    \State Assign level $j$ to every element of $X_j\setminus X_{j-1}$ \label{line:online-refreshlevels}
  \ElsIf{the event is the arrival of an element $e$}
    \State $B^+\gets\OPT(X_j+e)$ \label{line:online-greedytest}
    \If{$e\notin B^+$}
      \State Reject $e$
    \Else
      \State Set $F_i\gets\{x\in X_i:x\succ e\}$ for $0\le i\le j$, and set $L\gets\{x\in X_j:e\succ x\}$
      \State $r\gets\Insert(L,F,e,c|_L,A\cap L)$ \label{line:online-insert}
      \If{$r$ is failure}
        \State Reject $e$ and leave the reference sets unchanged
      \Else
        \State Apply the levels in $r$ (set $f\in X_i$ iff $r_f\le i$ for every $f\in L$), add $e$ to every $X_i$, and set $A\gets A+e$ \label{line:online-accept}
      \EndIf
    \EndIf
  \Else \Comment{the clock of some $f\in C_t\setminus A$ rings}
    \State Set $F_i\gets\{x\in X_i:x\succ f\}$ for $0\le i\le j$, and set $L\gets\{x\in X_j:f\succ x\}$
    \State $r\gets\Remove(L,F,f,c|_L,A\cap L)$ \label{line:online-remove}
    \State Apply $r$ to $X_0,\ldots,X_j$ (set $x\in X_i$ iff $r_x\le i$ for every $x\in L$), and remove $f$ from every $X_i$ \label{line:online-delete}
  \EndIf
  \If{the reference sets or $A$ changed during this event}
    \State Recompute $C_t$ and assign fresh independent clocks to $C_t\setminus A$ as in line~\ref{line:online-clocks} \label{line:online-restart}
  \EndIf
\EndWhile
\State \Return $A$
\end{algorithmic}
\end{algorithm}

\subsection{Initialization and refreshes}\label{sec:initrefresh}

Lines~\ref{line:online-sample} and~\ref{line:online-initialize} initialize the algorithm.  Reject every element with auxiliary arrival time at most $p=T_0$.  Let this sample be $S$.  Set
\(
X_0=S\) and 
\(A=\varnothing,
\)
and start stage $0$ at time $p$.

Lines~\ref{line:online-refresh} and~\ref{line:online-refreshlevels} describe a refresh.  At a refresh time $T_j$ with $j\ge1$, keep all old reference sets and create the new outer reference set
\begin{equation}\label{eq:refresh}
X_j=H_{T_j},
\end{equation}
where $H_{T_j}$ is the set of all elements that have arrived by that time.  Thus every element rejected or internally deleted since the previous refresh enters the new outer level.  

Recall that an element is active when it belongs to at least one current reference set.  Since
\(
X_0\subseteq X_1\subseteq\cdots\subseteq X_j,
\)
the active elements during stage $j$ are exactly the elements of $X_j$.  The level of an active element is the smallest index $i$ such that it belongs to $X_i$.

\subsection{Processing a real arrival}\label{sec:arrival}

Lines~\ref{line:online-greedytest}--\ref{line:online-accept} describe the processing of a real arrival.  Suppose $e$ arrives at time $t\in[T_j,T_{j+1})$.  First compute
\(
B^+=\OPT(X_j+e).
\)
If $e\notin B^+$, reject it and leave all reference sets unchanged.

Assume $e\in B^+$.  The elements larger than $e$ in each $X_i$ form the initial background $F_i$.  Let $L$ be the active elements smaller than $e$, together with their current levels and information about which of them have been accepted.  Run
\[
\Insert(L,F,e,c|_L,A\cap L).
\]
If the call fails, reject $e$ and leave the reference sets unchanged.  If it succeeds, apply the returned level changes to the reference sets: for every $f\in L$, adjust its memberships so that $f\in X_i$ exactly when $r_f\le i$.  Then put $e$ into every reference set $X_0,\ldots,X_j$ and accept it:
\(
A\leftarrow A+e.
\)

The greedy test performed before calling $\Insert$ is important.  The local insertion scan only redistributes the reference levels of already active smaller elements.  It does not change the outer reference set $X_j$ except for adding $e$: every element whose level is reassigned already belongs to $X_j$, and every returned level lies in $\{0,\ldots,j\}$, so the element still belongs to $X_j$ after the reassignment.  Hence the test $e\in\OPT(X_j+e)$ remains valid after the scan.

\subsection{Removing rejected reference elements}\label{sec:deletions}

Lines~\ref{line:online-clocks}, \ref{line:online-remove}, and~\ref{line:online-delete} summarize the internal-deletion rule.  To make room for future arrivals, the algorithm sometimes removes an already rejected element from the reference sets.  We call such a removal an \emph{internal deletion} because it changes only the auxiliary reference sets and does not revise the element's online rejection.

Define the set of elements in $X_0$ that also belong to $\OPT(X_j)$ by
\[
C_t=X_0\cap\OPT(X_j).
\]
Only unaccepted elements of $C_t$ receive internal deletion clocks.  Each $f\in C_t\setminus A$ has instantaneous deletion rate
\begin{equation}\label{eq:deletionrate}
\delta(t)=\frac{q(t)}{p-g(t)}.
\end{equation}
This means that, conditional on the clock of $f$ not having rung before time $t$, it rings during $[t,t+h]$ with probability $\delta(t)h+o(h)$ as $h\downarrow0$.   If the clock of $f$ is the next event, the algorithm performs an internal deletion for $f$.  It keeps all larger active elements fixed and runs
\[
\Remove(L,F,f,c|_L,A\cap L)
\]
on the smaller active elements, and then removes $f$ from every current reference set $X_0,\ldots,X_j$.  The element remains rejected and is now outside the outer reference set $X_j$.  The accepted set $A$ does not change.

Although the rate $\delta(t)$ varies with time, the next clock ring can still be sampled explicitly.  If the reference sets and $A$ do not change on an interval $[s,u]$ inside one stage, then
\[
\int_s^u\delta(v)\,dv
=
\log\frac{p-g(s)}{p-g(u)},
\]
because $g'=q$.  At the current time $s$, let $d=|C_s\setminus A|$.  If $d>0$, sample the first ring among these $d$ active clocks by drawing one random variable $E\sim\operatorname{Exp}(1)$ and solving
\[
d\int_s^u\delta(v)\,dv=E
\]
for the candidate ringing time $u$.  Equivalently,
\[
p-g(u)=\bigl(p-g(s)\bigr)e^{-E/d},
\]
so formula~\eqref{eq:qgstage} gives $u$ explicitly within the current stage.  Since all $d$ clocks have the same rate, conditional on a clock ringing next, its element is chosen uniformly from $C_s\setminus A$.  A change to the reference sets or $A$ may change $C_t\setminus A$, so the algorithm recomputes the elements that receive clocks.  Restarting the remaining clocks is valid because, conditional on no ring before time $t$, the probability that a clock also does not ring during $[t,u]$ is $\exp(-\int_t^u\delta(v)\,dv)$, the same as for a fresh clock started at $t$.

\subsection{Feasibility}

The local insertion and removal scans were designed to preserve $A\subseteq\OPT(X_0)$.  The following lemma guarantees this property throughout the online process.

\begin{lemma}[Feasibility invariant]\label{lem:feasible}
After initialization and after every arrival, refresh, and internal deletion,
\[
A\subseteq\OPT(X_0).
\]
In particular, the algorithm always returns an independent set.  An element is added to $A$ only at its real arrival and is never removed from $A$.
\end{lemma}

\begin{proof}
The claim holds initially.  A refresh changes only the reference hierarchy and leaves $X_0,A$ unchanged.

Consider an internal deletion of an element $z$, and let $X_i^-$ and $X_i^+$ denote the reference sets immediately before and after the deletion.  The algorithm chooses $z\in C_t\setminus A$, so $z$ is not accepted.  When $\Remove$ processes the largest element $f$ in the current set $L$ and $f\in A_L$, validity gives $c'_f=0$, and the procedure keeps $f$ at level $0$.  The only possible change to the level of a processed element is from level $0$ to a later level, in the case $c'_f=0$ and $f\notin A_L$.  Thus $\Remove$ may move unaccepted elements out of $X_0$, but it never moves an accepted element out of $X_0$ or any element into $X_0$.  The same holds in every recursive call.  Therefore, after applying the returned assignment and deleting the unaccepted element $z$, the accepted set is unchanged and
\(
A\subseteq X_0^+\subseteq X_0^-.
\)
Before the deletion, the induction hypothesis for the online process gives $A\subseteq\OPT(X_0^-)$.  Since $X_0^+$ is a subset of $X_0^-$ that still contains every element of $A$, Lemma~\ref{lem:inherit} gives $A\subseteq\OPT(X_0^+)$.  Thus an internal deletion preserves the invariant.

Next, consider a successful insertion of a new element $e$.  Let $X_i^-$ and $A^-$ denote the reference sets and accepted set immediately before the insertion, 
and let $X_i^+$ and $A^+$ denote them immediately afterward.  We prove that, after accepting $e$, the invariant \(
A^+\subseteq\OPT(X_0^+)
\) still holds.

First, we prove by induction on $|L|$ that every successful call $\Insert(L,F,z,c,A_L)$, after applying its output and adding $z$, keeps all elements of $A_L$ in the level-$0$ greedy set.  
If $L=\varnothing$, the claim is true.  If $\sigma_F(z)=0$, then $z\in\cl(F_0)$ and hence adding $z$ changes none of the closures $\cl(F_i)$, so the unchanged assignment $c$ still satisfies the claim.

Now let $f$ be the largest element of a nonempty $L$.  If $f\in A_L$, validity gives $c_f=0$.  Write $a=\sigma_F(f)$ and $b=\sigma_{F+z}(f)$.  If $a=b$, adding $z$ does not change the span status of $f$.  If $a>b$, case~(I1) applies.  The only way for $f$ to leave the level-$0$ greedy set is to have $b=0$.  In that case $g_b=g_0=g(T_0)=g(p)=0$, so the continuation probability $\frac{q_ag_b}{q_bg_a}$ in case~(I1) is zero and the call fails.  Hence, whenever the call succeeds, $f$ remains in the level-$0$ greedy set.  The recursive call on $L-f$ succeeds as well, so the induction hypothesis protects every accepted element in $A_L\setminus\{f\}$.  Its level changes concern only elements smaller than $f$ and therefore cannot change the greedy status of $f$.

If $f\notin A_L$, then $f$ is not accepted, so it remains only to protect the accepted elements in $L-f$.  If the procedure makes only one recursive call, that call is on $L-f$.  Since $f\notin A_L$, we have $A_L\setminus\{f\}=A_L$, and the induction hypothesis proves the claim for this case.  It remains to consider the recycling case~(I3), in which the procedure makes two recursive calls.  If the original call succeeds, both recursive calls succeed.  Applying the induction hypothesis first to the insertion of $z$ and then to the insertion of $f$ shows that all elements of $A_L$ remain in the level-$0$ greedy set after both updates.  This proves the claim.  In the initial call, $A_L=A^-\cap L$.  Every element of $A^-\setminus L$ is larger than the newly arriving $e$ and is unaffected by changes to smaller elements.  Consequently, every successful insertion keeps all elements of $A^-$ in $\OPT(X_0^+)$.

Second, we show that the newly accepted element $e$ also belongs to $\OPT(X_0^+)$.  The initial greedy test gives $e\in\OPT(X_j^-+e)$.  The insertion scan changes only the inner reference sets, so after the insertion
\(
X_j^+=X_j^-+e.
\)
Moreover, $e\in X_0^+\subseteq X_j^+$.  Lemma~\ref{lem:inherit} therefore gives $e\in\OPT(X_0^+)$.  Since $A^+=A^-+e$, the invariant follows.
\end{proof}

\section{Joint Distribution}\label{sec:distribution}

We now determine the exact joint distribution of the reference sets, the arrived and unarrived elements, and the accepted set.  

\subsection{State weights}

Fix a stage $j$ and an interior time $t\in(T_j,T_{j+1})$.  Let
\(
Y=X_j
\)
be the outer reference set, let $Z$ be the arrived elements outside $Y$, and let $U$ be the unarrived elements.  Then
\[
E=Y\mathbin{\dot\cup}Z\mathbin{\dot\cup}U.
\]
Each $f\in Y$ has a level $c_f\in\{0,\ldots,j\}$, and these labels determine all nested reference sets through
\[
X_i=\{f\in Y:c_f\le i\}.
\]

The tuple $(Y,Z,U,c,A)$ records the complete configuration at time $t$, and we call it a \emph{state}.  A state is valid when $A\subseteq Y$, every element of $A$ has level $0$, and $A\subseteq\OPT(X_0)$.  All other states are invalid.

The quantity $W_t(Y,c,A)$ is the contribution of the active configuration $(Y,c,A)$ to the weight of a complete state.   We now define $W_t(Y,c,A)$ explicitly.  For each $f\in Y$ and each $0\le i\le j$, let
\[
F_i(f)=\{x\in Y:x\succ f,\ c_x\le i\}.
\]
Thus $F_i(f)$ consists of the elements larger than $f$ that belong to $X_i$.  Write $F(f)=(F_0(f),\ldots,F_j(f))$.  The span index of $f$ is $\sigma_{F(f)}(f)=\min\{k\in\{0,\ldots,j\}:f\in\cl(F_k(f))\}$,
 and in particular, $\sigma_{F(f)}(f)=\staridx$ when $f\notin\cl(F_j(f))$. Write $g_{\staridx}=g(t)$.  The local factor of $f$ is
\[
w_t(f;Y,c,A)=
\begin{cases}
g_{\sigma_{F(f)}(f)}, & c_f=0,\ f\in A,\\
p-g_{\sigma_{F(f)}(f)}, & c_f=0,\ f\notin A,\\
d_{c_f}, & c_f\ge1,\ f\notin A,\\
0, & c_f\ge1,\ f\in A.
\end{cases}
\]
For a valid state, set
\begin{equation}\label{eq:stateweight}
W_t(Y,c,A)=\prod_{f\in Y}w_t(f;Y,c,A),
\end{equation}
and set $W_t(Y,c,A)=0$ for an invalid state.  This product is exactly the result of applying the recursive local-weight definition~\eqref{eq:localweight} to the elements of $Y$ in decreasing order, starting from the empty background.

Define
\begin{equation}\label{eq:jointdistribution}
\mu_t(Y,Z,U,c,A)
=
W_t(Y,c,A)
(t-T_j)^{|Z|}(1-t)^{|U|}.
\end{equation}
The next theorem shows that this formula gives the actual distribution of the online process.

\begin{theorem}[Multilevel thinning distribution]\label{thm:jointdistribution}
For every stage $j$, every $t\in[T_j,T_{j+1}]$, and every valid state $(Y,z,u,c,a)$,
\[
\Pr[(X_j,Z,U,c,A)=(Y,z,u,c,a)]
=
\mu_t(Y,z,u,c,a),
\]
At a refresh time $T_i$, the formula for stage $i-1$ is interpreted by the left limit $t\uparrow T_i$, immediately before the refresh, whereas the formula for stage $i$ is interpreted by the right limit $t\downarrow T_i$, immediately after the refresh.

Consequently, conditional on the reference sets $X_0,\ldots,X_j$, the events that individual elements of $\OPT(X_0)$ have been accepted are independent.  An element of $\OPT(X_0)$ first excluded from a greedy set at level $i$ has conditional acceptance probability $g_i/p$.  An element of $\OPT(X_0)$ still in $\OPT(X_j)$ at time $t$ has conditional acceptance probability $g(t)/p$.
\end{theorem}

Before proving the theorem, we check that~\eqref{eq:jointdistribution} is a probability distribution and that refreshes connect adjacent stages correctly.

\begin{lemma}[Normalization]\label{lem:globalnorm}
For every $t$ in stage $j$, the values in~\eqref{eq:jointdistribution} sum to one over all states.
\end{lemma}

\begin{proof}
Fix the active set $Y$.  Lemma~\ref{lem:localnorm}, with empty initial background, gives
\[
\sum_{c,A}W_t(Y,c,A)=T_j^{|Y|}.
\]
Therefore summing over all partitions $(Y,Z,U)$ gives
\[
\sum_{Y,Z,U}
T_j^{|Y|}(t-T_j)^{|Z|}(1-t)^{|U|}
=
\bigl(T_j+(t-T_j)+(1-t)\bigr)^n=1.
\]
\end{proof}

\begin{lemma}[Distribution at a refresh]\label{lem:refresh}
Suppose Theorem~\ref{thm:jointdistribution} holds through the end of stage $j-1$.  Applying the deterministic refresh $X_j=H_{T_j}$ transforms the distribution at the end of stage $j-1$ exactly into the initial distribution~\eqref{eq:jointdistribution} for stage $j$.
\end{lemma}

\begin{proof}
Immediately before the refresh, every element in $Z$ contributes the factor
\[
T_j-T_{j-1}=d_j.
\]
The refresh moves exactly these elements into the new outer level $j$, whose local factor is also $d_j$.

It remains to consider level-$0$ elements whose greedy status changes because the new outer level has been added.  Just before the refresh, every level-$0$ element still in the outer greedy set uses the current factor $g(T_j)$ if accepted and $p-g(T_j)$ if unaccepted.  If the new level makes such an element leave the greedy set, the factor used from that level onward is $g_j=g(T_j)$, the same value.  Elements that had already left at an earlier level keep their old factors.  Hence the probability assigned to every state is preserved exactly.
\end{proof}

\begin{proof}[Proof of Theorem~\ref{thm:jointdistribution}]
At $t=p$, the set $X_0$ is the Bernoulli-$p$ sample, $A=\varnothing$, and $Z=\varnothing$.  Since $g(p)=0$,~\eqref{eq:jointdistribution} gives exactly this initial distribution.

Fix an interior time $t$ of stage $j$ and a valid state
\[
\xi=(Y,Z,U,c,A).
\]
Let
\[
C=X_0\cap\OPT(X_j).
\]
Only the thinning factors of elements in $C$ depend on the current time.  Every accepted element in $A\cap C$ contributes $g(t)$, while every unaccepted element in $C\setminus A$ contributes $p-g(t)$.  Differentiating~\eqref{eq:jointdistribution} and using $g'=q$ gives
\begin{equation}\label{eq:logderivative}
\frac{\mu_t'(\xi)}{\mu_t(\xi)}
=
\frac{|Z|}{t-T_j}
-
\frac{|U|}{1-t}
+
\frac{|A\cap C|q(t)}{g(t)}
-
\frac{|C\setminus A|q(t)}{p-g(t)}.
\end{equation}
We show that these four terms are exactly the incoming and outgoing probability rates of the online process.

\paragraph{Arrival outflow.}
Each $e\in U$ has conditional arrival rate $1/(1-t)$.  Whether accepted or rejected, its arrival changes the state.  The total outgoing arrival rate is therefore
\(
\mu_t(\xi)\frac{|U|}{1-t},
\)
which gives the second term of~\eqref{eq:logderivative}.

\paragraph{Internal-deletion outflow.}
Each $f\in C\setminus A$ has deletion rate $q(t)/(p-g(t))$.  Hence the total deletion outflow is
\[
\mu_t(\xi)|C\setminus A|\frac{q(t)}{p-g(t)},
\]
which gives the fourth term.

\paragraph{Acceptance inflow.}
Fix $e\in A\cap C$.  Immediately before accepting $e$, it is unarrived.  Remove $e$ from the fixed state and keep the larger active elements unchanged.  The possible level assignments of the smaller elements immediately before the arrival are precisely the inputs to the insertion scan.  Since $e$ is in the current outer greedy set, its span index with respect to larger active elements is $\staridx$.  Equation~\eqref{eq:columnflow} gives a total insertion weight equal to
\(
q(t)
\)
times the weight of the smaller elements after insertion.  The fixed state additionally contains the accepted level-$0$ factor $g(t)$ for $e$.  Before the arrival, the unarrived element contributes $1-t$, which cancels the arrival rate $1/(1-t)$.  Thus every $e\in A\cap C$ contributes
\(
\mu_t(\xi)\frac{q(t)}{g(t)}
\)
to the acceptance inflow.  Summing gives the third term.

\paragraph{Rejection and deletion inflow.}
Fix $e\in Z$.  There is a unique state that can reach $\xi$ when $e$ arrives and is rejected without changing the reference sets:
\[
\xi_e^-=(Y,Z-e,U+e,c,A).
\]
Let $P_e$ be the conditional probability that, starting from $\xi_e^-$, the arriving element $e$ passes the greedy test and the top-level call to $\Insert$ succeeds.  Thus $P_e$ includes all random choices made by the insertion procedure; set $P_e=0$ if $e$ fails the greedy test.

The weight of $\xi_e^-$ is obtained from that of $\xi$ by replacing one factor $t-T_j$ by $1-t$.  After multiplication by the arrival rate $1/(1-t)$ and the rejection probability $1-P_e$, the direct rejection inflow is therefore
\[
\mu_t(\xi)\frac{1-P_e}{t-T_j}.
\]

It remains to sum the deletion inflow.  For every possible successful output assignment $c'$ of $\Insert$ from $\xi_e^-$, consider the state in which $e$ is active, unaccepted, and at level $0$, the smaller active elements have assignment $c'$, and all other parts of the state agree with $\xi$.  These are exactly the possible predecessor states that can reach $\xi$ by internally deleting $e$.  In each such state, $e\in X_0\cap\OPT(X_j)$ and its local factor is $p-g(t)$.  This factor cancels its deletion rate $q(t)/(p-g(t))$.

If $e$ fails the greedy test, no state with $e\in X_0\cap\OPT(X_j)$ can produce this deletion, so both $P_e$ and the deletion inflow are zero.  We may therefore assume below that $e$ passes the greedy test.  In that case $e$ is not spanned by the larger active elements, so $\sigma_F(e)=\staridx$.

Let $F,L,c|_L$, and $A\cap L$ be the inputs determined by $\xi_e^-$ as in the arrival rule.  Summing \eqref{eq:detailedflow} over all output assignments $c'$ gives
\[
\sum_{c'}W_F(L,c|_L,A\cap L)I_F(c|_L,c')
=
q_{\sigma_F(e)}
\sum_{c'}W_{F+e}(L,c',A\cap L)D_F(c',c|_L).
\]
The left side is $W_F(L,c|_L,A\cap L)P_e$, and here $q_{\sigma_F(e)}=q(t)$.  Including the factors of the larger active elements and the elements outside the outer reference set, the total deletion inflow is therefore
\[
\mu_t(\xi)\frac{P_e}{t-T_j}.
\]
Hence the rejection and deletion inflows associated with $e$ add to
\[
\frac{\mu_t(\xi)}{t-T_j}.
\]
Summing over $e\in Z$ gives the first term of~\eqref{eq:logderivative}.

All possible state changes are arrivals, internal deletions, or deterministic refreshes.  We have shown that the candidate probabilities satisfy the full forward equations in the interior of each stage.  On every compact subinterval away from time $1$, the state space is finite and the total transition rate is bounded.  Hence the corresponding finite linear differential system has a unique solution.  The initial condition and Lemma~\ref{lem:refresh} identify~\eqref{eq:jointdistribution} with the actual distribution through every stage.  The statement at time $1$ follows by taking the left limit.

It remains to derive the conditional thinning statement.  Fix the reference sets $X_0,\ldots,X_j$.  They determine $Y=X_j$, the level assignment $c$, $\OPT(X_0)$, and all span indices $\sigma_{F(f)}(f)$.  Fix $B\subseteq\OPT(X_0)$ and evaluate the product $W_t(Y,c,A)$ in~\eqref{eq:stateweight} at $A=B$.  The only factors in this product that vary with the choice of $B$ are
\begin{equation}\label{eq:conditionalfactors}
\prod_{f\in B}g_{\sigma_{F(f)}(f)}
\prod_{f\in\OPT(X_0)\setminus B}
\bigl(p-g_{\sigma_{F(f)}(f)}\bigr).
\end{equation}
Every element of $X_0\setminus\OPT(X_0)$ is forced to be unaccepted and contributes the fixed factor $p-g_0=p$, while every element at a later level contributes its fixed factor $d_i$.  The partition of $E\setminus X_j$ into $Z$ and $U$ is not fixed by the reference sets, but summing over all such partitions contributes the common factor
\[
\sum_{Z\mathbin{\dot\cup}U=E\setminus X_j}
(t-T_j)^{|Z|}(1-t)^{|U|}
=(1-T_j)^{|E\setminus X_j|},
\]
which is also independent of $B$.  Hence all factors other than those in~\eqref{eq:conditionalfactors} cancel when conditioning on the reference sets.  Summing~\eqref{eq:conditionalfactors} over all $B\subseteq\OPT(X_0)$ gives
\[
\begin{aligned}
&\sum_{B\subseteq\OPT(X_0)}
\prod_{f\in B}g_{\sigma_{F(f)}(f)}
\prod_{f\in\OPT(X_0)\setminus B}
\bigl(p-g_{\sigma_{F(f)}(f)}\bigr)\\
&\qquad=
\prod_{f\in\OPT(X_0)}
\left(g_{\sigma_{F(f)}(f)}+p-g_{\sigma_{F(f)}(f)}\right)\\
&\qquad=p^{|\OPT(X_0)|}.
\end{aligned}
\]
The conditional probability of $A=B$ is the weight in~\eqref{eq:conditionalfactors} for this particular $B$, divided by the sum of that weight over all possible subsets.  The preceding calculation shows that this denominator is $p^{|\OPT(X_0)|}$.  Therefore
\[
\begin{aligned}
\Pr[A=B\mid X_0,\ldots,X_j]
&=
\frac{
\displaystyle
\prod_{f\in B}g_{\sigma_{F(f)}(f)}
\prod_{f\in\OPT(X_0)\setminus B}
\bigl(p-g_{\sigma_{F(f)}(f)}\bigr)
}{p^{|\OPT(X_0)|}}\\
&=
\prod_{f\in B}\frac{g_{\sigma_{F(f)}(f)}}p
\prod_{f\in\OPT(X_0)\setminus B}
\frac{p-g_{\sigma_{F(f)}(f)}}p\\
&=
\prod_{f\in B}\frac{g_{\sigma_{F(f)}(f)}}p
\prod_{f\in\OPT(X_0)\setminus B}
\left(1-\frac{g_{\sigma_{F(f)}(f)}}p\right).
\end{aligned}
\]
This is a product of Bernoulli probabilities, so the events $\{f\in A\}$ for $f\in\OPT(X_0)$ are conditionally independent, and the conditional acceptance probability of each $f$ is $g_{\sigma_{F(f)}(f)}/p$.  Suppose $f$ first leaves the greedy set at level $i$.  Then $f\notin\cl(F_k(f))$ for every $k<i$, whereas $f\in\cl(F_i(f))$, so $\sigma_{F(f)}(f)=i$.  Substituting this span index into the Bernoulli parameter gives
\[
\Pr[f\in A\mid X_0,\ldots,X_j]=\frac{g_i}{p}.
\]
If $f$ remains in $\OPT(X_j)$, then no $F_k(f)$ spans it, so $\sigma_{F(f)}(f)=\staridx$.  Since $g_{\staridx}=g(t)$, its conditional acceptance probability is $g(t)/p$.
\end{proof}

The theorem is the exact multilevel analogue of the joint distribution in Dynamic Thinning.  The  difference is that the original single parameter $1-q(t)$ is replaced by a conditional acceptance probability determined by the first level at which an element leaves the greedy set.

\section{Competitive Ratio and Query Complexity}\label{sec:ratio}

The joint distribution reduces the competitive analysis to calculating the acceptance probability of one optimal element.

\begin{proposition}[Acceptance probability for a canonical-optimum element]\label{prop:selection}
For every refresh schedule satisfying $p\ge1/e$ and every $e\in O$,
\(
\Pr[e\in A]=g(1).
\)
Consequently, Multilevel Dynamic Thinning is $1/g(1)$-probability-competitive.
\end{proposition}

\begin{proof}
By Lemma~\ref{lem:inherit}, whenever $e$ belongs to a reference set $X_i$, it also belongs to $\OPT(X_i)$. 
 Hence a canonical-optimum element at level $0$ never leaves any outer greedy set.  For every realization of the final reference hierarchy with $e\in X_0$, 
 Theorem~\ref{thm:jointdistribution} gives the conditional acceptance probability
\[
\Pr[e\in A\mid X_0,\ldots,X_{m-1}]=\frac{g(1)}p.
\]
Since the right-hand side is the same for every such realization, we have
\[
\begin{aligned}
\Pr[e\in A\mid e\in X_0]
&=
\mathbb E\!\left[
\Pr[e\in A\mid X_0,\ldots,X_{m-1}]
\,\middle|\,e\in X_0
\right]=\frac{g(1)}p.
\end{aligned}
\]

We next compute the probability that $e$ has level $0$ by marginalizing the joint distribution in Theorem~\ref{thm:jointdistribution} over all other elements and over the acceptance status of $e$.  Recall that a ``factor'' here means the multiplicative contribution of $e$ to the state weight in~\eqref{eq:jointdistribution}.  If $e$ has level $0$, its local factor is $g_{\sigma_{F(e)}(e)}$ when $e$ is accepted and $p-g_{\sigma_{F(e)}(e)}$ when it is unaccepted.  Thus, after summing over these two possible acceptance states, its total factor is
\[
g_{\sigma_{F(e)}(e)}+
\bigl(p-g_{\sigma_{F(e)}(e)}\bigr)=p.
\]
If $e$ has level $i\in\{1,\ldots,m-1\}$, its local factor is $d_i$.  It remains to identify the factor when $e$ is not active.  
At the final time, use the notation of Theorem~\ref{thm:jointdistribution} with $j=m-1$ and $t=1$.  Then $Y=X_{m-1}$ and $U=\varnothing$,  so an element $e\notin X_{m-1}$ belongs to $Z$.  The term $(t-T_j)^{|Z|}$ in~\eqref{eq:jointdistribution} assigns one factor $t-T_j$ to each element of $Z$. Hence $e$ contributes
\(
t-T_j=1-T_{m-1}.
\) The resulting marginal weights for the mutually exclusive possible locations of $e$ sum to
\[
p+\sum_{i=1}^{m-1}d_i+(1-T_{m-1})
=p+(T_{m-1}-p)+(1-T_{m-1})=1.
\]
They are therefore the marginal probabilities of the possible locations of $e$, and in particular
\(
\Pr[e\in X_0]=p.
\)
Consequently,
\[
\Pr[e\in A]
=
\Pr[e\in X_0]\Pr[e\in A\mid e\in X_0]
=
p\frac{g(1)}p
=g(1).
\]
\end{proof}

\subsection{A geometric multilevel schedule}

We now choose the refresh times.  For an integer $m\ge1$, set
\begin{equation}\label{eq:geomtimes}
r=1+\frac1m,
\qquad
T_i=r^{i-m}
\quad(0\le i\le m).
\end{equation}
Thus $T_m=1$ and
\(
p=T_0=r^{-m}>e^{-1}.
\)
Every stage has the same relative length, that is,
\(
\frac{T_{i+1}-T_i}{T_i}=\frac1m.
\)
By~\eqref{eq:qgstage}, we have
\(
q(T_i)=e^{-i/m}.
\)
Integrating $q$ stage by stage gives an exact formula.

\begin{lemma}[Geometric acceptance probability]\label{lem:Gm}
For the schedule~\eqref{eq:geomtimes},
\begin{equation}\label{eq:Gm}
G_m:=g(1)
=
r^{-m}\bigl(1-e^{-1/m}\bigr)
\sum_{i=0}^{m-1}\left(re^{-1/m}\right)^i.
\end{equation}
Moreover,
\begin{equation}\label{eq:Gmbound}
\frac1{G_m}<e\left(1+\frac1m\right).
\end{equation}
\end{lemma}

\begin{proof}
The contribution of stage $i$ to $g(1)$ is
\[
\int_{T_i}^{T_{i+1}}q(t)\,dt
=T_iq(T_i)\bigl(1-e^{-1/m}\bigr).
\]
Substituting $T_i=r^{i-m}$ and $q(T_i)=e^{-i/m}$ and summing proves~\eqref{eq:Gm}.

For the bound, note that for $t\in[T_i,T_{i+1})$ we have
\(
T_i\ge \frac tr.
\)
Hence
\[
\int_p^t\frac{ds}{b(s)}
\le
r\int_p^t\frac{ds}{s}
=r\log\frac tp,
\]
and therefore
\(
q(t)\ge\left(\frac pt\right)^r.
\)
It follows that
\[
G_m\ge
\int_p^1\left(\frac pt\right)^r dt
=
\frac{p}{r-1}\left(1-p^{r-1}\right).
\]
Since $r-1=1/m$ and $p^{1/m}=r^{-1}$, the right side equals $p/r=r^{-(m+1)}$.  Thus
\[
\frac1{G_m}\le r^{m+1}
<e\,r
=e\left(1+\frac1m\right),
\]
where $(1+1/m)^m<e$.
\end{proof}

We now prove the competitive ratio in Theorem~\ref{thm:main}.  

\begin{proof}[Proof of Theorem~\ref{thm:main}]
Choose $m=\lceil e/\eps\rceil$ and use the geometric schedule~\eqref{eq:geomtimes}.
Lemma~\ref{lem:Gm} gives
\[
\frac1{G_m}<e\left(1+\frac1m\right)
=e+\frac em
\le e+\eps.
\]
Proposition~\ref{prop:selection} therefore proves the probability-competitive guarantee.  The query complexity follows from Theorem~\ref{thm:complexity} below.
\end{proof}

For illustration, Table~\ref{tab:ratios} lists several values for different numbers $m$ of stages.  The one-stage geometric choice fixes $p=1/2$ and is $3.163953$-probability-competitive.
This is slightly worse than the $3.1462$-guarantee of the one-stage Dynamic Thinning, for which the  parameter is optimized to be $p_*\approx0.46594$.

\begin{table}[t]
\centering
\caption{Multilevel guarantees.  The number of stages is $m$ and the number of refreshes is $m-1$.}
\label{tab:ratios}
\begin{tabular}{rcc}
\toprule
$m$ & $p=(1+1/m)^{-m}$ & $1/G_m=1/g(1)$ \\
\midrule
1  & 0.500000 & 3.163953 \\
2  & 0.444444 & 2.994226 \\
4  & 0.409600 & 2.871416 \\
8  & 0.389744 & 2.798927 \\
16 & 0.379085 & 2.759664 \\
32 & 0.373554 & 2.739244 \\
\bottomrule
\end{tabular}
\end{table}

\subsection{Query complexity}\label{sec:complexity}

We next show that for every fixed number of stages, the number of independence queries is quadratic in $n$ in the worst case.

At stage $j$, span indices lie in
\[
\{0,1,\ldots,j,\staridx\}.
\]
View all recursive calls generated by one use of $\Insert$ or $\Remove$ as a tree.  The initial call is the root, each recursive call made by another call is its child, and a call making no further call is a leaf.  At each node, consider the span index of the element being inserted or removed by that call.  A call with one child never increases this index.  The recycling step can give a node two children.  By Lemma~\ref{lem:layerexchange}, the indices at both children are strictly smaller than the index at their parent whenever the second child occurs with positive probability.  Such a split requires a later level $i\ge1$.  Therefore every path from the root to a leaf contains at most $j$ splits.

\begin{lemma}[Recursion size]\label{lem:recursionsize}
At stage $j$, a call to $\Insert$ or $\Remove$ on at most $n$ elements visits $O(2^j n)$ recursive nodes for every outcome of its random choices.
\end{lemma}

\begin{proof}
There are at most $2^j$ leaves because the span index can strictly decrease through at most $j$ positive level values at splits.  Between two splits, each recursive step removes the largest remaining element from the scan, so every path from the root to a leaf has length at most $n+1$.  Multiplying gives the bound.
\end{proof}

To implement span-index calculations, maintain for each background level $F_i$ a greedy set.  Testing whether $z\in\cl(F_i)$ requires one independence query on that greedy set plus $z$.  When a recursive step adds one processed element to $F_i,F_{i+1},\ldots,F_j$, each affected greedy set can be updated with one query.  Thus each recursive node uses $O(j+1)$ independence queries.  Combining this bound with the $O(2^j n)$ recursive nodes from Lemma~\ref{lem:recursionsize}, each call to $\Insert$ or $\Remove$ made by the online algorithm during stage $j$ uses
\begin{equation}\label{eq:localeffort}
O\bigl((j+1)2^j n\bigr)
\end{equation}
independence queries.

The next lemma bounds the number of calls made by the online algorithm.

\begin{lemma}[At most one deletion per element per stage]\label{lem:oneperstage}
During a fixed stage $j$, each element of $E$ can be selected for at most one internal deletion.  Hence there are at most $n$ internal deletions in that stage and at most $mn$ over the whole execution.
\end{lemma}

\begin{proof}
When an element is selected for an internal deletion, it is removed from the current outer reference set $X_j$ and placed among the arrived elements outside $X_j$.  During the remainder of the same stage, insertion and removal scans operate only on elements already in $X_j$ and on a newly arriving element.  They never access an old element outside $X_j$.  The deleted element can become active again only at the next scheduled refresh, when a new outer reference set is created from the set of arrived elements.
\end{proof}


\begin{theorem}[Worst-case query complexity]\label{thm:complexity}
Multilevel Dynamic Thinning with $m$ stages uses
\(
O(m2^m n^2)
\)
independence queries in the worst case.  
\end{theorem}

\begin{proof}
We count separately the top-level calls to $\Remove$, the top-level calls to $\Insert$, and the queries used to maintain the greedy sets.  At stage $j$,~\eqref{eq:localeffort} gives
\(
O((j+1)2^j n)
\)
queries per top-level call made by the online algorithm.  By Lemma~\ref{lem:oneperstage}, at most $n$ internal-deletion events occur in stage $j$, and each such event makes one top-level call to $\Remove$.  The total number of queries in these calls is therefore
\[
O(
  n^2\sum_{j=0}^{m-1}(j+1)2^j
)
=O(m2^m n^2).
\]

Each of the $n$ real arrivals invokes $\Insert$ at most once, namely when it passes the greedy test.  If this call occurs during stage $j$, then \eqref{eq:localeffort} bounds its number of queries by
\[
O\bigl((j+1)2^j n\bigr)
\le O\bigl(m2^{m-1}n\bigr).
\]

Finally, the greedy test at an arrival and the recomputation of $\OPT(X_j)$ and $C_t$ after a state-changing arrival or internal deletion each require $O(n)$ additional independence queries.  There are at most $n$ arrivals and, by Lemma~\ref{lem:oneperstage}, at most $mn$ internal deletions, so these maintenance queries total $O(mn^2)$, absorbed by $O(m2^m n^2)$.
\end{proof}

This complexity is a worst-case statement over both arrival orders and the algorithm's internal random choices.  The notation $O_\eps(n^2)$ in Theorem~\ref{thm:main} means that the hidden constant may depend on the fixed accuracy parameter.  Our explicit implementation gives a dependence of the form $2^{O(1/\eps)}$.

\section{Relation to Prior Work}\label{sec:related}

Multilevel Dynamic Thinning is a direct extension of Dynamic Thinning by Chan et al.~\cite{clw2026}: it replaces one reference set by a nested hierarchy, and the case $m=1$ recovers their algorithm with the same sampling parameter $p$.

\begin{proposition}[Dynamic Thinning as the one-level case]\label{prop:basecase}
Suppose $m=1$, so there is no refresh and $b(t)=p$.  Then Multilevel Dynamic Thinning reduces exactly, on the pair $(X_0,A)$, to Dynamic Thinning with parameter $p$ from~\cite{clw2026}.
\end{proposition}

\begin{proof}
Equation~\eqref{eq:qgdef} gives
\[
q(t)=e^{-(t-p)/p},
\qquad
g(t)=p(1-q(t)).
\]
There are no levels beyond $0$, hence no recycling step in the insertion rule of Section~\ref{sec:local}.  If a newly arriving element $e$ does not belong to $\OPT(X_0+e)$, the algorithm rejects it.  If $e$ increases the reference rank, the insertion scan reaches the case $L=\varnothing$ and succeeds with probability $q(t)$.  If $e$ exchanges with an unaccepted member of $\OPT(X_0)$, the scan succeeds with probability one.  If the exchange partner has already been accepted, the relevant continuation probability is zero.  These are exactly the four acceptance cases of Dynamic Thinning.

The internal deletion rate in~\eqref{eq:deletionrate} is
\[
\frac{q(t)}{p-g(t)}=\frac{q(t)}{pq(t)}=\frac1p.
\]
With no later level, the paired removal scan changes no other reference-set membership.  Thus an internal event simply deletes an unaccepted member of $\OPT(X_0)$ at rate $1/p$, as in Dynamic Thinning.
\end{proof}

Both algorithms build on Singla's framework of reversibly updating a random reference configuration while preserving its distribution~\cite{singla2026}.
The more specific conditional-thinning invariant comes from Abdi et al.'s single-sample prophet algorithm~\cite{abdi2026}.  Conditional on the current reference, each already processed element in the reference greedy set has been accepted independently with probability $1/2$.  In Dynamic Thinning \cite{clw2026}, conditional on the current reference, each element in the reference greedy set has been accepted independently with the time-dependent probability $1-q(t)$.  Their joint distribution for the reference set, the arrival partition, and the accepted set is the direct predecessor of our Theorem~\ref{thm:jointdistribution}.  Our theorem extends that distribution to several reference levels, with conditional acceptance probabilities determined by span indices.

The main technical challenge in our multilevel construction is to design the insertion and removal rules and pair their recursive choices so that the prescribed joint distribution is maintained.  
These procedures and the method of pairing their recursive choices are derived directly from Huang's forward and reverse kernels~\cite{huang2026}.
His construction pairs successful forward transitions with reverse transitions, preserves the survivor set, and maintains a uniform latent vector through a probability-mass balance.
Huang \cite{huang2026} explicitly attributes his choice to control the full output distribution to Singla \cite{singla2026} and to the stationary online contention resolution schemes of Aminian,
 Niazadeh, and Nuti~\cite{aminian2026}.  The latter work \cite{aminian2026} designs distributions on feasible sets and implements them online through updates preserving those distributions.  

Our algorithm is a finite-level discretization of Huang's survival-process construction.  For each fixed instance, under the geometric refresh schedule~\eqref{eq:geomtimes}, the distribution of the final accepted set converges to that of Huang's algorithm with sampling threshold $1/e$ as $m\to\infty$.  
This follows because the discrete random time defined by the reference hierarchy converges to Huang's continuous random time, while the corresponding conditional acceptance probabilities converge to his survival probabilities.  Huang attains the exact $e$ guarantee with expected polynomial query complexity.  Our finite hierarchy yields $e+\eps$ for each fixed $\eps>0$ and a worst-case $O_\eps(n^2)$ query bound.  The dependence on $\eps$ is exponential in $1/\eps$ in our implementation, so this is a different ratio--complexity tradeoff.

\paragraph{AI statement.}
The author used GPT-5.6 Sol and GPT-6 Astra to help design the algorithm and develop its proofs.  The author has checked the proofs and takes full responsibility for their correctness and for the contents of this paper.

\end{document}